\newif\ifllncs
 \llncstrue

\ifllncs
\let\accentvec\vec 
\documentclass[a4paper]{llncs}
\let\vec\accentvec 
\else
\fi

\usepackage[T1]{fontenc} 
\usepackage[utf8]{inputenc}
\usepackage[american]{babel} 

\usepackage[cmex10]{amsmath}

\usepackage{amsfonts}
\usepackage{amssymb}
\usepackage{graphicx}
\usepackage{color}
\usepackage{longtable}
\usepackage{pstricks-add}
\usepackage{pdflscape}
\usepackage{import}
\usepackage[T1]{fontenc} 
\subimport{}{ADT.sty}
\usepackage[colorlinks]{hyperref}
\usepackage{breakurl}
\usepackage{eurosym}
\usepackage{siunitx}
\usepackage{graphicx}
\usepackage[lmargin=2.5cm,rmargin=2.5cm,tmargin=3.5cm,bmargin=3.5cm]{geometry}

\definecolor{lightblue}{rgb}{0.5,0.5,1.0}
\definecolor{darkred}{rgb}{0.5,0,0}
\definecolor{darkgreen}{rgb}{0,0.5,0}
\definecolor{darkblue}{rgb}{0,0,0.5}

\hypersetup{colorlinks=true
,linkcolor=darkred
,filecolor=darkgreen
,urlcolor=darkred
,citecolor=darkblue}

\ifllncs
\spnewtheorem*{Proof}{Proof}{\itshape}{\rmfamily}
\renewenvironment{proof}{\begin{Proof}}{\qed\end{Proof}}
\else
\usepackage{amsthm}
\theoremstyle{plain}

\newtheorem{theorem}{Theorem}

\theoremstyle{definition}
\newtheorem{example}{Example}
\newtheorem{definition}[theorem]{Definition}

\fi

\newcommand{\BK}[1]{\textcolor{blue}{BK: #1}}
\newcommand{\PS}[1]{\textcolor{violet}{[#1]}}
\newcommand{\SM}[1]{\textcolor{magenta}{SM: #1}}

\newcommand{\comment}[1]{#1}
\renewcommand{\BK}[1]{}
\renewcommand{\PS}[1]{}
\renewcommand{\SM}[1]{}
\renewcommand{\comment}[1]{}

\newcommand{\reals}{\mathbb{R}}
\newcommand{\nats}{\mathbb{N}}
\newcommand{\getz}{\mathbb{R}_{\geq 0}}

\newcommand{\vp}{\vee^{\pro}}

\newcommand{\vo}{\vee^{\opp}}
\newcommand{\wo}{\wedge^{\opp}}

\newcommand{\negate}[1]{\overline{#1}}

\DeclareMathOperator{\energy}{erg}
\DeclareMathOperator{\pb}{pb}
\DeclareMathOperator{\cost}{co}
\DeclareMathOperator{\sat}{sat}
\DeclareMathOperator{\pro}{p}
\DeclareMathOperator{\owner}{own}
\DeclareMathOperator{\opp}{o}
\DeclareMathOperator{\counter}{c}

\DeclareMathOperator{\avg}{avg}

\DeclareMathOperator{\rnk}{rank}
\DeclareMathOperator{\arity}{in}
\DeclareMathOperator{\type}{out}

\title{Quantitative Questions on Attack--Defense Trees}

\author{
Barbara Kordy, Sjouke Mauw and Patrick Schweitzer
}
\institute{
University of Luxembourg, SnT\\ 
{\tt \{barbara.kordy, sjouke.mauw, patrick.schweitzer\}@uni.lu}
}

\begin{document}

\maketitle

%
\BK{
\begin{enumerate}
\item \label{rel-work} 
Rewrite related work: an idea would be to decrease the number of 
mentioned papers/approaches, but clearly state how our solution addresses the 
problems mentioned in the mentioned ones and how it relates to them.
\end{enumerate}
}

\begin{abstract}
Attack--defense trees are a novel methodology for graphical security
modeling and assessment. The methodology includes visual, intuitive 
tree models whose analysis is supported by a rigorous mathematical 
formalism. Both, the intuitive and the formal components of the 
approach can be used for quantitative analysis of attack--defense scenarios. 
In practice, we use intuitive questions to ask about aspects of scenarios we 
are interested in. Formally, a computational procedure, defined with the help 
of attribute domains and a bottom-up algorithm, is applied to derive the 
corresponding numerical values. 

This paper bridges the gap between the intuitive and the formal way 
of quantitatively assessing attack--defense scenarios. 
We discuss how to properly specify a question, so that it can be answered 
unambiguously. Given a well specified question, we then show how to 
derive an appropriate attribute domain which 
constitutes the corresponding formal model.
Since any attack tree is in particular an attack--defense tree, our analysis 
is also an advancement of the attack tree methodology.
\end{abstract}



\section{Introduction}
\label{sec:introduction}
In graphical security modeling, the main focus lies on the visual 
representation of a scenario. A common requirement of graphical models is 
their user friendliness, and hence their intuitiveness. However, intuitive 
models are prone to be ambiguous. While this in itself may already be 
undesirable, ambiguity is detrimental for computer supported processing. 
Contrary to intuitive models, formal frameworks prevent ambiguity and are 
able 
to support automated quantitative evaluation. A disadvantage of formal 
frameworks is however that they are, not seldom, more difficult to 
understand. 

Attack--defense trees~\cite{KoMaRaSc2} form a systematic, graphical 
methodology for analysis of  attack--defense scenarios. 
They represent a game between an attacker, whose goal 
is to attack a system, and a defender who tries to protect the system.
The widespread formalism of attack trees is a subclass of 
attack--defense 
trees, where only the actions of the attacker are considered.
The attack--defense tree methodology combines intuitive and formal components. 
On the one hand, the intuitive visual attack--defense tree representation is 
used in practice to answer qualitative and quantitative questions, such as 
``What are the minimal costs to protect a server?'', or ``Is the scenario 
satisfiable?'' On the other hand, there exist attack--defense terms and a 
precise mathematical framework for quantitative analysis using 
a recursive bottom-up procedure introduced for attack trees in~\cite{Schn} 
and extended to attack--defense trees in~\cite{KoMaRaSc}.

Several case studies performed using the attack--defense tree methodology 
showed that there exist a significant discrepancy between users focusing on 
the intuitive 
components of the model and users working with the formal components. 
This is due to the fact that intuitive models are 
user friendly but often ambiguous. 
In contrast, formal models are rigorous and mathematically sound.
This, however, makes them difficult 
to understand for users without a formal background.
This discrepancy between the two worlds is especially visible 
in the case of quantitative analysis. 
Correct numerical evaluation can only be performed when 
all users have precise and consistent understanding of 
considered quantities also called attributes.

\noindent\emph{\textbf{Contributions.}} 
This work is an attempt to bridge the gap between the
intuitive and the formal components of the attack--defense tree methodology
for quantitative security analysis. Our goal is to provide 
a precise relation between intuitive questions and their formal 
models called attribute domains. 
We elaborate which kind of intuitive questions 
occurring in practical security analysis 
can be answered with the help of the bottom-up procedure on attack--defense 
trees. We empirically classify questions that were collected during case 
studies and literature reviews. 
We distinguish and formally analyze three different classes of 
questions: 
those referring to one player, those where answers for both players 
can be deduced from each other and those relating to an outside third 
party.  
For each class we provide detailed guidelines 
how the questions should be specified, so that they are unambiguous and can 
be answered correctly. 
Simultaneously, we discuss templates of the attribute domains 
corresponding to each class.

\noindent\emph{\textbf{Related work.}} 
An excellent historical overview on graphical security modeling, starting 
from fault trees~\cite{VeGoRoHa}, over threat trees~\cite{Amor} and privilege 
graphs~\cite{DaDe} leading up to Schneier's attack trees~\cite{Schn},
was given by Pi\`{e}tre-Cambac\'{e}d\`{e}s and Bouissou in~\cite{PiBo}.
When Schneier  introduced the attack trees formalism in~\cite{Schn},
he proposed how to evaluate, amongst others, attack costs, success 
probability of an attack, and whether there is a need for special equipment. 
Since then, many authors have not only extended the attack tree 
formalism syntactically, but also followed in his footsteps and included the 
possibility of quantitative analysis in their extended formalisms. 
Baca and Petersen~\cite{BaPe}, for example, have extended attack trees to 
countermeasure graphs and quantitatively analyzed an open-source application 
development. Bistarelli~et~al.~\cite{BiDaPe}, 
Edge~et~al.~\cite{EdDaRaMi} and Roy~et~al.~\cite{RoKiTr2} have augmented 
attack trees with a notion of defense or mitigation nodes. They all analyze 
specific types of risk using particular risk formulas, adjusted 
to their models. Willemson and J\"{u}rgenson~\cite{WiJu} introduced an order 
on the leaves of attack trees to be able to optimize the 
computation of the expected outcome of the attacker.
%
There also exist a number of case studies and experience reports that 
quantitatively analyze real-life systems. Notable examples are 
Henniger~et~al.~\cite{HeApFuRoRuWe}, who have conducted a study using attack 
trees for vehicular communications systems, Abdulla~et~al.~\cite{AbCeKa}, who 
analyzed the GSM radio network using attack jungles, and Tanu and 
Arreymbi~\cite{TaJo}, who 
assessed the security of mobile SCADA system for a tank and pump facility. 
Since all previously mentioned papers focus on specific attributes, they do 
not address the general problem of the relation between intuitive and formal 
quantitative analysis.
%
%

The formalism of attack--defense trees considered in this work 
was introduced by Kordy et al. in~\cite{KoMaRaSc}. 
Formal aspects of the attack--defense methodology have been discussed 
in~\cite{KoMaMeSc} and~\cite{KoPoSc}.
In~\cite{BaKoMeSc12}, Bagnato et al. provided guidelines for 
how to use attack--defense trees in practice. They analyzed a DoS attack 
scenario on an RFID-based goods management system by evaluating a number 
of relevant attributes, including cost, time, detectability, penalty,
required skill level, impact, difficulty and profitability.

\noindent\emph{\textbf{Paper structure.}} 
The necessary background concerning the attack--defense tree methodology 
is briefly explained in Section~\ref{sec:attributes}. The relation between 
intuitive and formal quantitative analysis of attack--defense scenarios is 
presented in~Section~\ref{sec:first_classification}. This section also 
introduces our classification of questions that can be answered on 
attack--defense trees with the help of a bottom-up procedure. The  
classification contains three classes of questions which are treated in 
Sections~\ref{sec:two-op-case},~\ref{sec:winning} and~\ref{sec:environmental}.
Section~\ref{sec:tool} presents a software tool, that has been developed to 
support quantitative analysis of attack--defense scenarios. 
Section~\ref{sec:conclusion} concludes the paper.

\section{Attack--Defense Scenarios Intuitively and Formally}
\label{sec:attributes}

\subsection{The Intuitive Model}
\label{sec:intuitive_model}

An \emph{attack--defense tree} (ADTree) 
constitutes an intuitive graphical model
describing the measures an attacker might take 
in order to attack a system and the defenses that a defender can employ to 
protect the system. 
An ADTree is a node-labeled rooted tree having 
nodes of two opposite types: 
\emph{attack nodes} represented with circles and \emph{defense nodes}
represented with rectangles.
The root node of an ADTree depicts the main goal of one of the 
players. 
Each node of an ADTree may have one or more children of the
same type which \emph{refine} the node's goal into subgoals. 
The refinement relation is indicated by solid edges and can be
either disjunctive or conjunctive.
The goal of a \emph{disjunctively} refined node is achieved when \emph{at 
least one} of its children's goals is achieved. The goal of a 
\emph{conjunctively} refined node
is achieved when \emph{all} of its children's goals are achieved. 
To distinguish between the two refinements we indicate 
the conjunctive refinement with an arc. 
A node which does not have any children of the same type is called a 
\emph{non-refined} node. Non-refined nodes represent \emph{basic actions}, 
i.e., actions which can be easily understood and quantified.
Every node in an ADTree may also have one child of the opposite type, 
representing a \emph{countermeasure}. The countermeasure relation is indicated 
by dotted edges. Nodes representing countermeasures can again be refined into 
subgoals and countered by a node of the opposite type. 
\begin{example}\label{ex:main}
An example of an ADTree is given in~Figure~\ref{fig:main}. The root of the 
tree represents an attack on a server. Three ways to accomplish this
attack are depicted: insider attack, outsider attack (OA) and 
stealing the server (SS). 
To achieve his goal, an insider needs to be 
internally connected (IC) and have the correct user credentials 
(UC). To not be caught easily, an insider uses 
a colleague's and not his own credentials.
Attack by an outsider can be prevented if a properly configured firewall (FW) is 
installed.
\end{example}
\begin{figure}
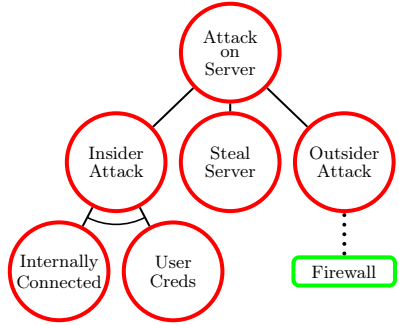

\centering
\settowidth{\pbs}{MIIII}
%
%
%
%
\begin{minipage}[c]{5.7cm}
\setlength{\pbs}{1.5cm}
\psscalebox{0.75}{
\pstree[levelsep=55pt,treesep=0.2cm]{\NodeA{Attack\\on\\Server}}{
  \pstree{\NNodeA{A}{Insider\\Attack}}{
    \NNodeA{B}{Internally\\Connected}
    \NNodeA{C}{User\\Creds}
    }
  \NodeA{Steal\\Server}
  \pstree{\NodeA{Outsider\\Attack}}{
    \NodeBC{Firewall}
    }
  }
\Arca{1.1cm}{A}{B}{C}
}
\end{minipage}
%
%
%
\caption{An ADTree for how to attack a server}
\label{fig:main}
\end{figure}
Graphical visualization of potential attacks and possible countermeasures 
constitutes a first step towards a systematic security analysis. The next step 
is to assign numerical values to ADTree models, i.e., to perform a 
quantitative analysis.
Intuitively speaking, performing a quantitative security analysis means 
\emph{answering questions related to specific aspects or properties 
influencing the security of a system or a company}. These questions may be of 
Boolean type, e.g., ``Is the attack satisfiable?'', or may concern physical or 
temporal aspects, e.g., ``What are the minimal costs of attacking a system?'',
or ``How long does it take to detect the attack?'' 
In order to facilitate and automate the analysis of vulnerability 
scenarios using ADTrees, the formal model of ADTerms and their quantitative 
analysis have been introduced. We briefly describe them in the next section.


\subsection{The Formal Model}
\label{sec:formal_model}
In this section we recall formal definitions related to our methodology.
For more details and  explanatory examples we refer the reader 
to~\cite{KoMaRaSc2}.
To formally represent and analyze ADTrees, 
typed terms over a particular typed 
signature, called the AD--signature, have been introduced 
in~\cite{KoMaRaSc}. 
To be able to capture ADTrees rooted in an  
attacker's node as well as those rooted in a defender's node, we distinguish 
between the \emph{proponent} (denoted by~$\pro$), which refers to the root 
player, and the \emph{opponent} (denoted by~$\opp$), which is the other 
player. 
For instance, for the ADTree in Figure~\ref{fig:main}, the proponent is the 
attacker and the opponent is the defender. Conversely, if the root of an 
ADTree is a defense node, the proponent is the defender and the opponent is 
the attacker. 

Furthermore, given a set~$\mathcal{S}$, we denote by~$\mathcal{S}^*$ the set 
of all finite 
strings over~$\mathcal{S}$, and by~$\varepsilon$ the empty string. 
For $s\in\mathcal{S}$, we denote by $s^+$ a string composed of 
a finite number of symbols $s$.
\begin{definition}\label{AD--signature}
The \emph{AD--signature} is a pair~$\Sigma=(\mathcal{S},\mathcal{F})$, where
\begin{itemize}
\item $\mathcal{S} = \{\pro, \opp \}$ is a set of types, and 
\item $\mathcal{F}=\mathbb{B}^{\pro}\,\cup\, 
\mathbb{B}^{\opp}\,\cup\,\{\vee^{\pro},\wedge^{\pro},
\vee^{\opp}, \wedge^{\opp}, \counter^{\pro}, \counter^{\opp}\}$ is a set of 
function symbols,
such that the sets~$\mathbb{B}^{\pro}$, $\mathbb{B}^{\opp}$ 
and~$\{\vee^{\pro},\wedge^{\pro},
\vee^{\pro}, \wedge^{\opp}, \wedge^{\opp}, \counter^{\pro}, 
\counter^{\opp}\}$ are pairwise disjoint. 
\end{itemize}
Every function symbol~$F\in\mathcal{F}$ is equipped with a mapping~$\rnk\colon
\mathcal{F}\to \mathcal{S}^*\times \mathcal{S}$, where 
$\rnk(F)$ is defined as a pair $(\arity(F),\type(F))$. The first 
component of the pair describes the type of the arguments of $F$ and the second 
component describes the type of the values of $F$. We have 
\begin{align*} 
&\rnk(b)=(\varepsilon,\pro), \text{\ for\ } b\in\mathbb{B}^{\pro}, && 
\rnk(b)=(\varepsilon,\opp), \text{\ for\ } b\in\mathbb{B}^{\opp},\\
&\rnk(\vee^{\pro})=(\pro^+,\pro), 
&&\rnk(\vee^{\opp})=(\opp^+,\opp), \\
&\rnk(\wedge^{\pro})=(\pro^+,\pro),
&&\rnk(\wedge^{\opp})=(\opp^+,\opp),\\
&\rnk(\counter^{\pro})=(\pro\opp,\pro),
&&\rnk(\counter^{\opp})=(\opp\pro,\opp).
\end{align*}
\end{definition}
Given 
$F\in\mathcal{F}$ and $s\in\mathcal{S}$, we say that 
$F$ is of type $s$, if $\type(F)=s$. 
The elements of~$\mathbb{B}^{\pro}$ and~$\mathbb{B}^{\opp}$ are typed 
constants, which represent basic actions of the proponent's 
and opponent's type, respectively. By~$\mathbb{B}$ we denote
the union~$\mathbb{B}^{\pro}\cup \mathbb{B}^{\opp}$. The 
functions\footnote{In fact, symbols $\vee^{\pro},\wedge^{\pro},\vee^{\opp}$,
and~$\wedge^{\opp}$ represent unranked functions, i.e.,
they stand for families of functions
$(\vee^{\pro}_k)_{k\in\mathbb{N}},
(\wedge^{\pro}_k)_{k\in\mathbb{N}}, (\vee^{\opp}_k)_{k\in\mathbb{N}}, 
(\wedge^{\opp}_k)_{k\in\mathbb{N}}$.
}
$\vee^{\pro},\wedge^{\pro},\vee^{\opp}$, and~$\wedge^{\opp}$ 
represent disjunctive and conjunctive refinement operators for the
proponent and the opponent, respectively. We set~$\negate{\pro}=\opp$ 
and~$\negate{\opp}= \pro$. The binary functions~$\counter^s$, for~$s\in 
\mathcal{S}$, represent countermeasures and are used to connect components of 
type~$s$ with components of the opposite type~$\negate{s}$.
%
%
%
\begin{definition}\label{de:adterm}
Typed ground terms over the AD--signature~$\Sigma$ are called attack--defense 
terms (ADTerms). 
The set of all ADTerms is denoted by~$\mathbb{T}_{\Sigma}$.
\end{definition}
For~$s\in\{\pro,\opp\}$, we denote by~$\mathbb{T}_\Sigma^{s}$ the set of all 
ADTerms
with the head symbol of type~$s$. We 
have~$\mathbb{T}_\Sigma=\mathbb{T}_{\Sigma}^{\pro}\cup
\mathbb{T}_{\Sigma}^{\opp}$. The elements of~$\mathbb{T}_\Sigma^{\pro}$ 
and~$\mathbb{T}_\Sigma^{\opp}$
are called \emph{ADTerms of the proponent's} and \emph{of the opponent's 
type}, respectively. The ADTerms of the proponent's type constitute formal
representations of ADTrees. 
\begin{example}
\label{ex:main_term}
Consider the ADTree given in Figure~\ref{fig:main}. The corresponding 
ADTerm  
is
\[t=\vee^{\pro}(\wedge^{\pro}(\text{IC},\text{UC}),\text{SS},\counter^{\pro}
(\text{OA},\text{FW})).\] 
The entire ADTerm, as well as its six 
subterms~$\wedge^{\pro}(\text{IC},\text{UC})$,~$\counter^{\pro}(\text{OA},
\text{FW})$,~IC,~UC,~SS,
and~OA, are of the proponent's type. Term $t$ also contains a subterm 
of the opponent's type, namely~FW.
\end{example}

In order to facilitate and automate quantitative analysis of vulnerability 
scenarios, the notion of an attribute
for ADTerms has been formalized in~\cite{KoMaRaSc}. 
An attribute expresses a particular property, quality, or characteristic of a 
scenario, such as the minimal costs of an attack or the expected impact of a 
defensive measure. 
A specific bottom-up procedure for evaluation of attribute values on ADTerms
ensures that the user, for instance a security expert, only needs to quantify 
the basic actions. From these, the value for the entire scenario is deduced 
automatically. Attributes are formally modeled using attribute domains.
\begin{definition}\label{def:attribute-domain}
An attribute domain for ADTerms is a tuple 
\[
A_{\alpha}=(D_{\alpha}, \vee^{\pro}_{\alpha}, \wedge^{\pro}_{\alpha},
\vee^{\opp}_{\alpha}, \wedge^{\opp}_{\alpha}, \counter^{\pro}_{\alpha}, 
\counter^{\opp}_{\alpha}),\]
where~$D_{\alpha}$ is a set of values and, for~$s\in\{\pro,\opp\}$,
\begin{itemize}
\item $\vee^{s}_{\alpha}$,~$\wedge^{s}_{\alpha}$ are unranked operations 
on~$D_{\alpha}$, 
\item $\counter^s$ are binary operations on~$D_{\alpha}$. 
\end{itemize}
\end{definition}

Let~$A_{\alpha}=(D_{\alpha}, \vee^{\pro}_{\alpha}, \wedge^{\pro}_{\alpha},
\vee^{\opp}_{\alpha}, \wedge^{\opp}_{\alpha}, \counter^{\pro}_{\alpha}, 
\counter^{\opp}_{\alpha})$ be an attribute domain for ADTerms. The bottom-up 
computation of attribute values on ADTerms is formalized as follows. 
First, a value from $D_{\alpha}$ is assigned to each basic action, 
with the help of function $\beta_{\alpha}\colon \mathbb{B}\to D_{\alpha}$, 
called a \emph{basic assignment}. 
Then, a recursively defined function~$\alpha\colon \mathbb{T}_\Sigma \to D_{\alpha}$ 
assigns a value to every ADTerm~$t$, as follows 
\begin{equation}\label{eq:alpha(v)}
\alpha(t)= \begin{cases} 
\beta_{\alpha}(t), &\text{if~$t\in \mathbb{B}$,} \\
\vee^s_{\alpha}(\alpha(t_1),\dots, \alpha(t_k)), 
&\text{if~$t=\vee^s(t_1,\dots,t_k)$},\\
\wedge^s_{\alpha}(\alpha(t_1),\dots, \alpha(t_k)), 
&\text{if~$t=\wedge^s(t_1,\dots,t_k)$},\\
\counter^s_{\alpha}(\alpha(t_1),\alpha(t_2)), 
&\text{if~$t=\counter^s(t_1,t_2)$},
\end{cases}
\end{equation}
where~$s\in\{\pro, \opp\}$ and~$k>0$. 

The example below illustrates the bottom-up procedure 
for an attribute called  \emph{satisfiability}.
\begin{example}\label{ex:sat}
The question ``Is the considered scenario satisfiable?'' 
is formally modeled using the satisfiability attribute. 
The corresponding attribute domain is 
$A_{\sat}=(\{0,1\},\vee,\wedge, \vee,\wedge, 
\star,\star)$, where~$\star(x,y) =x\wedge\neg y$, for all~$x,y\in\{0,1\}$. 
The basic assignment~$\beta_{\sat}\colon \mathbb{B}\to \{0,1\}$ assigns the
value $1$ to every basic action which is 
satisfiable and the value $0$ to every basic action which is
not satisfiable. 
Using the recursive evaluation procedure defined by 
Equation~\eqref{eq:alpha(v)}, we evaluate the satisfiability attribute 
on the ADTerm from Example~\ref{ex:main_term}. Assuming that all basic actions
are satisfied, i.e., that 
$\beta_{\sat}(X)=1$  for $X\in\{\text{IC}, \text{UC},\text{SS},
\text{OA},\text{FW}\}$, 
we obtain
\begin{align*}
\sat&(\vee^{\pro}(\wedge^{\pro}(\text{IC},\text{UC}),\text{SS},\counter^{\pro}(\text{OA},\text{FW})))=\\
&\vee(\wedge(\beta_{\sat}(\text{IC}),\beta_{\sat}(\text{UC})),
\beta_{\sat}(\text{SS}),\star(\beta_{\sat}(\text{OA}),\beta_{\sat}(\text{FW})))=\\
&\vee(\wedge(1,1),1,\star(1,1))= \vee(1,1,0)=1.
\end{align*}
\end{example}

The satisfiability attribute, as introduced in the previous example, 
allows us to define which player is the \emph{winner} of the considered 
attack--defense scenario. If the satisfiability value calculated for an 
ADTerm is equal to~$1$, the winner of the corresponding scenario
is the proponent, otherwise the winner is the opponent. 
In Example~\ref{ex:sat}, the root attack is satisfied, 
so the winner is the attacker. 

\section{Classification of Questions}
\label{sec:first_classification}

One of the goals of this paper is to describe how to correctly specify a 
question for an ADTree. This allows us to construct the corresponding formal 
model and deduce an answer using the bottom-up procedure. Let us motivate our 
approach with the following example. 
\begin{example}\label{eg:imprecise}
``What are the costs of 
the considered scenario?'' seems to be a valid question on an ADTree. 
However, this question is 
underspecified, because we do not know whether we should quantify the 
attacker's costs, the defender's costs or both. Clarifying this information is 
necessary to correctly define the corresponding basic assignment. We improve 
the question and ask ``What are the costs of the attacker?'' The new question 
is still underspecified, since it is not clear whether we are interested in 
the minimal, maximal, average or other costs. Making also this information 
explicit is necessary to correctly define the way how to aggregate the values 
for disjunctively refined nodes of the attacker. 
\end{example}

In this paper, we provide a pragmatic taxonomy of 
quantitative questions that can be asked about ADTrees. 
The presented classification results from case studies, 
e.g.,~\cite{BaKoMeSc12,EdDaRaMi,TaJo}, 
as well as from a detailed literature overview concerning quantitative 
analysis of security. Our study allowed us to identify three main classes 
of empirical questions, as described below. 

\medskip

\noindent\emph{\textbf{Class~1: Questions referring to one player.}}
Most of the typical questions for ADTrees have an explicit or implicit 
reference to one of the players which we call \emph{owner} of the question. 
This is motivated by the fact that the security model is usually analyzed from 
the point of view of one player only.
Examples of questions referring to one player are
``What are the \emph{minimal costs of the attacker}?'' (here the owner is the 
attacker) 
or
``How much does it \emph{cost to protect} the system?'' (here implicitly 
mentioned owner is the defender). 
When we ask a question of Class~$1$, we assume that 
its owner does not have extensive information concerning his adversary. 
Thus, we always consider the worst case scenario
with respect to the actions of the other player. 
Most of the questions usually asked for attack trees
can be adapted so that they can be answered on ADTrees as well. 
Thus,  questions related to attributes such as
\emph{attacker's/defender's costs}
\cite{Schn,BuLaPrSaWi,TaJo,BaPe,SaDuPa,MaOo,Yage,AbCeKa,RoKiTr2,ByFrMi,Program1,WaWhPhPa,EdDaRaMi}, 
\emph{attack/defense time}~\cite{HeApFuRoRuWe,Schn,WaWhPhPa}, 
\emph{attack detectability}~\cite{TaJo,ByFrMi}, 
\emph{attacker's special skill}~\cite{MaOo,AbCeKa,Schn}, 
\emph{difficulty of attack/protection}
\cite{ByFrMi,FuChWaLeTaAnLi,TaJo,HeApFuRoRuWe,MaOo,AbCeKa,Amor,WaWhPhPa}, 
\emph{penalty}~\cite{BuLaPrSaWi,JuWi3,WaWhPhPa}, 
\emph{impact of the 
attack}~\cite{Schn,TaJo,HeApFuRoRuWe,LiLiFeHe,SaDuPa,MaOo,Amor,AbCeKa,RoKiTr,EdDaRaMi,WaWhPhPa},
 \emph{attacker's profit}~\cite{Amor,JuWi3,BiDaPe,RoKiTr2}, etc., 
all belong to Class~$1$. We analyze 
questions of this class in~Section~\ref{sec:two-op-case}.

\medskip

\noindent\emph{\textbf{Class~2: 
Questions where answers for both players can be deduced from each other.}}
Exemplary questions belonging to Class~$2$ are
``Is the \emph{scenario satisfiable}?'', or
``How \emph{probable is it that the scenario will succeed}?''.
We observe that if the scenario is satisfied for the attacker, then it is not 
satisfied for the defender, and vice versa. Similarly, 
knowing that one player succeeds with probability $p$, we also 
know that the other player succeeds with probability $1-p$. The foremost goal 
of attack trees and all their extensions is to represent whether attacks are 
possible. Thus, the \emph{satisfiability} attribute is considered, either 
explicitly or implicitly, in all works concerning attack trees and their 
extensions. As for \emph{probability}\footnote{
We would like to point out that, the probability attribute 
can only be evaluated using the bottom-up procedure given by 
Equation~\eqref{eq:alpha(v)}, 
if the ADTree does not contain any dependent actions.}, the attribute has been extensively 
studied 
in~\cite{Schn,BuLaPrSaWi,HeApFuRoRuWe,LiLiFeHe,MaThFe,Yage,AbCeKa,RoKiTr2,ByFrMi,EdDaRaMi,WaWhPhPa}.
We perform a detailed analysis of questions of Class~$2$ in 
Section~\ref{sec:winning}.

\medskip

\noindent\emph{\textbf{Class~3: Questions referring to an outside third 
party.}}
Questions belonging to Class~$3$ relate to a universal property which is 
influenced by actions of both the attacker and the defender. 
They quantify attack--defense scenarios from the point of view of 
an outside third party which is neither the attacker nor the defender. 
For instance, one could ask about ``How much \emph{data traffic is involved in 
the attack--defense scenario}?''. In this case, we do not need to distinguish 
between traffic resulting from the attacker's and the defender's actions,
as both players contribute to the total amount. Another example of a question of 
Class~$3$ is ``What is the \emph{global environmental impact} of the 
scenario?''. Instances of environmental impact could be 
$CO_2$ emission or water pollution.
Attributes corresponding to questions in Class~$3$ have not been 
addressed in the attack tree literature, since attack trees focus
on a single player.
The importance of those questions becomes apparent when actions of two 
opposite parties are considered. The case study~\cite{BaKoMeSc12} that we 
have performed using the attack--defense tree methodology showed  
that such attributes relate to essential properties which should not be 
disregarded by the security assessment process. Questions of Class~$3$ are  
discussed in Section~\ref{sec:environmental}.

The following three sections set up guidelines for how to correctly specify 
quantitative questions of all three classes.
The guidelines' main purpose is to enable us to find a corresponding  
attribute domain in order to correctly compute an answer using the 
bottom-up 
procedure. Figure~\ref{fig:types} depicts the three classes of 
questions, as well as general templates for the corresponding attribute 
domains, as introduced in Definition~\ref{def:attribute-domain}. 
Symbols $\bullet, \circ,\diamond$ and $\overline{\bullet}$
serve as placeholders for specific operators. 
Corresponding symbols within a tuple indicate that 
the functions 
coincide. For instance, $(D, \circ,\bullet, \bullet,\circ, 
\bullet,\circ)$ means that 
$\vee^{\pro}_{\alpha}=\wedge^{\opp}_{\alpha}=\counter^{\opp}_{\alpha}$
and 
that $\wedge^{\pro}_{\alpha}=\vee^{\opp}_{\alpha}=\counter^{\pro}_{\alpha}$. 
We motivate these equalities and give possible instantiations of  
$\bullet, \circ,\diamond$ and $\overline{\bullet}$
in the following three sections.  \PS{Repetition the following three sections.}

\begin{figure}[htb]
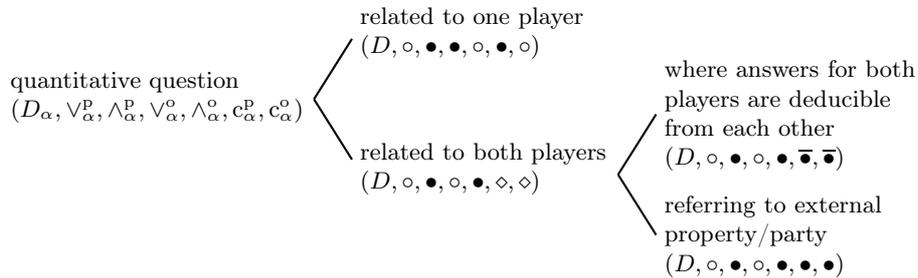

\centering
\vspace*{-0.5cm}
\hspace*{-0.3cm}
\pspicture*(-4.7,0.1)(7.6,3.8)
\put(-4.4, 2.45){
\parbox{3.9cm}{
quantitative question\newline
$(D_{\alpha}, \vee^{\pro}_{\alpha}, \wedge^{\pro}_{\alpha},
\vee^{\opp}_{\alpha}, \wedge^{\opp}_{\alpha}, \counter^{\pro}_{\alpha}, 
\counter^{\opp}_{\alpha})$
}
}

\psline(-0.3,2.5)(0.2,3.3)
\psline(-0.3,2.5)(0.2,1.7)
\put(0.2,3.3){
\parbox{3cm}{
related to one player
\newline
$(D, \circ,\bullet, \bullet,\circ, \bullet,\circ)$
}
}
\put(2.2,2){
\parbox{0.9cm}{
}
}
\put(0.2,1.5){
\parbox{3.3cm}{
related to both players\newline 
$(D, \circ,\bullet, \circ, 
\bullet,\diamond,\diamond)$
}
} 
\psline(3.7,1.5)(4.2,2.3)
\put(4.2,2.2){
\parbox{3.5cm}
{where answers for both\newline 
players are deducible\newline
from each other\newline
$(D, \circ,\bullet, \circ, 
\bullet,\overline{\bullet},\overline{\bullet})$
}
}
\psline(3.7,1.5)(4.2,0.7)
\put(4.2,0.6){
\parbox{3.1cm}
{referring to external\newline 
property/party\newline 
$(D, \circ,\bullet, \circ, 
\bullet,\bullet,\bullet)$}
}
\endpspicture

\vspace*{-0.4cm}
\caption{Classification of questions and  
attribute domains' templates}
\label{fig:types}
\end{figure}
%
%
%
%

\section{Questions Referring to One Player}
\label{sec:two-op-case}



\subsection{Defining a Formal Model for Questions of Class 1}
\label{sec:reduction_class_1}

Questions belonging to Class~$1$ refer to exactly one player, which we call 
the question's \emph{owner}. 
As we explain below, in the attack--defense tree setting, only two 
situations occur for a question's owner:
either he needs to choose \emph{at least one} option or he needs to execute 
\emph{all} 
options. Therefore, two operators suffice to answer questions 
of Class~$1$ and the generic attribute domain
is of the form~$(D, \circ,\bullet, \bullet,\circ, \bullet,\circ)$.
Furthermore, if we change a question's owner, the attribute domain
changes from $(D, \circ,\bullet, \bullet,\circ, \bullet,\circ)$ 
into $(D, \bullet, \circ, \circ, \bullet, \circ, \bullet)$.

We illustrate the construction of the formal model 
for Class~$1$
using the question ``What are the minimal costs of the attacker?'', 
where the owner is the attacker.  
In the case of Class~$1$, all values assigned to nodes and subtrees
express the property that we are interested in from the perspective of the 
question's owner. In the minimal costs example, this means that even 
subtrees rooted in defense nodes have to be quantified from the attacker's 
point of view, i.e., a value assigned to the root of a subtree 
expresses what is the minimal amount of money that the attacker needs to 
invest in order to be successful in the current subtree. 

Subtrees rooted in uncountered attacker's nodes can either be disjunctively 
or conjunctively refined. In the first case the attacker needs to ensure that 
he is successful in \emph{at least one} of the refining nodes, in 
the second case he needs to be successful in \emph{all} refining nodes. The 
situation for subtrees rooted in uncountered defender's nodes is reciprocal.
If a defender's node is disjunctively refined, the attacker needs to 
successfully counteract \emph{all} possible defenses to ensure that he is 
successful at the subtree's root node; if the defender's node is conjunctively 
refined, successfully counteracting \emph{at least one} of the refining nodes 
already suffices for the attacker to be successful at the subtree's root 
node. 

This reciprocality explains that two different operators suffice to quantify 
all possible uncountered trees: The operator that we use to combine attribute 
values for disjunctively refined nodes of one player is the same as the 
operator we use for conjunctively refined nodes of the other player.

Furthermore, the same two operators can also be used to quantify all remaining 
subtrees. If a subtree is rooted in a countered attacker's node, the
attacker needs 
to ensure that he is successful at the action represented by the root node 
\emph{and} that he successfully counteracts the existing defensive measure. 
Dually, for the attacker to be successful in a subtree rooted in a 
defender's countered node, it is sufficient to successfully overcome the defensive 
action \emph{or} to successfully perform the attack represented by the 
countering node.
This implies that we can use the same operator as for 
conjunctively refined attacker's nodes in the first case and the same operator 
as for disjunctively refined attacker's nodes in the second case. 


\subsection{Pruning}
\label{sec:pruning}
%
%
For attributes in Class~$1$, we are only interested in one player, 
the owner of a question. 
Therefore for this class, we should disregard 
subtrees that do not lead to a successful scenario for the owner. 
We achieve this with the help of the \emph{pruning} procedure 
illustrated in the following example.

\begin{example}
Consider the ADTree in Figure~\ref{fig:main} and assume that we are interested 
in calculating the minimal costs of the attacker. In this case, there is no 
need to consider the subtree rooted in ``Outsider Attack'', because it is 
countered by the defense ``Firewall'' and thus does not lead to a successful 
attack. The subtree rooted in ``Outsider Attack'' therefore should be removed. 
This simultaneously eliminates having to provide values for the non-refined 
nodes ``Outsider Attack'' and `Firewall''.
The computation of the minimal costs is then executed 
on the term corresponding to the tree in the
right of Figure~\ref{fig:for_pruning}. 
\end{example}
\begin{figure}[htb]
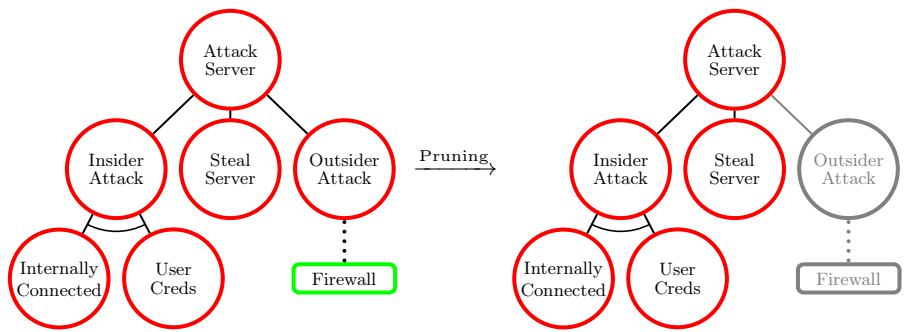

\centering
\begin{minipage}[c]{5.3cm}
\setlength{\pbs}{1.5cm}
\psscalebox{0.75}{
\pstree[levelsep=55pt,treesep=0.2cm]{\NodeA{Attack\\Server}}{
  \pstree{\NNodeA{A}{Insider\\Attack}}{
    \NNodeA{B}{Internally\\Connected}
    \NNodeA{C}{User\\Creds}
    }
  \NodeA{Steal\\Server}
  \pstree{\NodeA{Outsider\\Attack}}{
    \NodeBC{Firewall}
    }
  }
\Arca{1.1cm}{A}{B}{C}
}
\end{minipage}
$\xrightarrow{\text{Pruning}}$
\begin{minipage}[c]{5.3cm}
\setlength{\pbs}{1.5cm}
\psscalebox{0.75}{
\pstree[levelsep=55pt,treesep=0.2cm]{\NodeA{Attack\\Server}}{
  \pstree{\NNodeA{A}{Insider\\Attack}}{
    \NNodeA{B}{Internally\\Connected}
    \NNodeA{C}{User\\Creds}
    }
  \NodeA{Steal\\Server}
  \pstree{\NodeGAL{Outsider\\Attack}}{
    \NodeGBD{Firewall}
    }
  }
\Arca{1.1cm}{A}{B}{C}
}
\end{minipage}

\caption{Pruning the ``attack server'' scenario for 
questions of Class~$1$ owned by the attacker
}
\label{fig:for_pruning}
\end{figure}

To motivate the use of the pruning procedure, let us distinguish 
two situations. 
If a non-refined node of the non-owner is countered, its assigned 
value should not influence the result of the computation. 
If a non-owner's node is not countered, its value 
should indicate that the owner does not have a chance to successfully perform 
this subscenario. 
Mathematically, it means that 
the value assigned to the non-refined nodes of the non-owner 
needs to be neutral with respect to one operator and simultaneously absorbing 
with respect to the other. 
Since, in general, such an element may not exist, we use pruning 
to eliminate one of the described situations.  
which results in elimination of the absorption condition.

Below we explain how to intuitively prune an ADTree and how to model the 
pruning in a mathematical way. 

\medskip
\noindent\emph{\textbf{Pruning intuitively.}}
\label{para:intuitive_pruning}
Let us consider a question of Class~$1$ and its owner.
In order to graphically prune an ADTree, we perform the following procedure. 
Starting from a leaf of the non-owner, we traverse the tree towards 
the root until we reach the first node~$v$ satisfying one of the following 
conditions, as illustrated in 
Figures~\ref{fig:for_pruning_1},~\ref{fig:for_pruning_2},~\ref{fig:for_pruning_3},
and~\ref{fig:for_pruning_4}.

\begin{figure}[tbh]
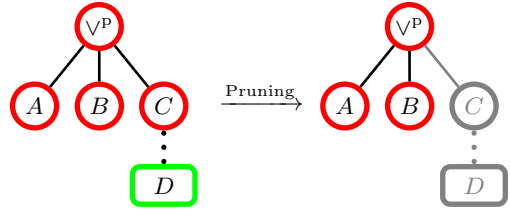

\centering

\settowidth{\pbs}{MII}
\begin{minipage}[c]{2.75cm}
\psscalebox{1}{
\pstree[levelsep=30pt,treesep=0.2cm]{\NodeA{$\vp$}}{
  \NodeA{$A$}
  \NodeA{$B$}
  \pstree{\NodeA{$C$}}{
    \NodeBC{$D$}
    }
  }
}
\end{minipage}
$\xrightarrow{\text{Pruning}}$
\begin{minipage}[c]{2.75cm}
\psscalebox{1}{
\pstree[levelsep=30pt,treesep=0.2cm]{\NodeA{$\vp$}}{
  \NodeA{$A$}
  \NodeA{$B$}
  \pstree{\NodeGAL{$C$}}{
    \NodeGBD{$D$}
    }
  }
}
\end{minipage}
\caption{Pruning a proper disjunctive refinement}
\label{fig:for_pruning_1}
\end{figure}

\begin{figure}[tbh]
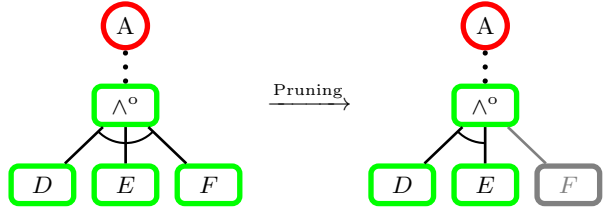

\centering

\settowidth{\pbs}{MII}
\begin{minipage}[c]{3.4cm}
\psscalebox{1}{
\pstree[levelsep=30pt,treesep=0.2cm]{\NodeA{A}}{
  \pstree{\NNodeBC{B}{$\wo$}}{
    \NNodeB{D}{$D$}
    \NNodeB{E}{$E$}
    \NNodeB{F}{$F$}
    }
  }
\Arca{0.5cm}{B}{D}{F}
}
\end{minipage}
$\xrightarrow{\text{Pruning}}$
\begin{minipage}[c]{3.4cm}
\psscalebox{1}{
\pstree[levelsep=30pt,treesep=0.2cm]{\NodeA{A}}{
  \pstree{\NNodeBC{B}{$\wo$}}{
    \NNodeB{D}{$D$}
    \NNodeB{E}{$E$}
    \NodeGBL{$F$}
    }
  }
\Arca{0.5cm}{B}{D}{E}
}
\end{minipage}

\caption{Pruning a proper conjunctive refinement}
\label{fig:for_pruning_2}
\end{figure}

\begin{figure}[tbh]
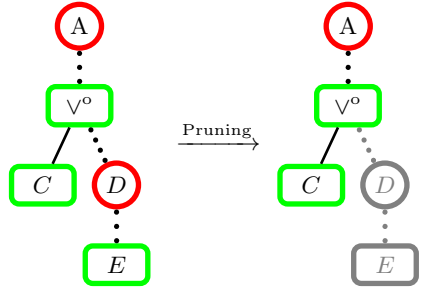

\centering

\settowidth{\pbs}{MII}
\begin{minipage}[c]{2.2cm}
\psscalebox{1}{
\pstree[levelsep=30pt,treesep=0.2cm]{\NodeA{A}}{
  \pstree{\NNodeBC{B}{$\vo$}}{
    \NNodeB{D}{$C$}
    \pstree{\NNodeAC{B}{$D$}}{
      \NodeBC{$E$}
      }
    }
  }
}
\end{minipage}
$\xrightarrow{\text{Pruning}}$
\begin{minipage}[c]{2.2cm}
\psscalebox{1}{
\pstree[levelsep=30pt,treesep=0.2cm]{\NodeA{A}}{
  \pstree{\NNodeBC{B}{$\vo$}}{
    \NodeB{$C$}
    \pstree{\NodeGAD{$D$}}{
      \NodeGBD{$E$}
      }
    }
  }
}
\end{minipage}

\caption{Pruning a countermeasure}
\label{fig:for_pruning_3}
\end{figure}

\begin{figure}[tbh]
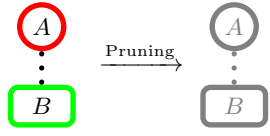

\centering

\settowidth{\pbs}{MII}
\begin{minipage}[c]{1.2cm}
\psscalebox{1}{
\pstree[levelsep=30pt,treesep=0.2cm]
			{\NNodeAC{B}{$A$}}{
      \NodeBC{$B$}
      }
    }
\end{minipage}
$\xrightarrow{\text{Pruning}}$
\begin{minipage}[c]{1.2cm}
\psscalebox{1}{
    \pstree[levelsep=30pt,treesep=0.2cm]{\NodeGAD{$A$}}{
      \NodeGBD{$B$}
      }
    }
\end{minipage}

\caption{Pruning an entire ADTree}
\label{fig:for_pruning_4}
\end{figure}
\begin{itemize}
\item $v$ is a node of the owner and part of a 
proper\footnote{A refinement is called proper if it contains at least two 
refining nodes.} 
disjunctive refinement (see Figure~\ref{fig:for_pruning});
\item $v$ is a node of the non-owner and part of a proper 
conjunctive refinement (see Figure~\ref{fig:for_pruning_2});
\item $v$ is a node of the owner that counteracts a refined node of the 
non-owner (see Figure~\ref{fig:for_pruning_3});
\item  $v$ is the root of the ADTree (see Figure~\ref{fig:for_pruning_4}).
\end{itemize}
The subtree rooted in node~$v$ is removed from the ADTree. The procedure is 
repeated, starting from all leaves of the non-owner. We note that 
the order in which 
we perform the procedure does not influence the final result. 
Also, in some cases the pruning procedure results in the removal of the 
entire ADTree. This is the case when the owner of the question 
does not have any way of successfully achieving his goal.

\medskip
\noindent\emph{\textbf{Pruning formally.}}
\label{sub:formal_pruning} 
Let $Q$  be a question~$Q$ of Class~$1$ and 
let~$\owner$ denotes the owner of $Q$. 
In order to model the 
pruning procedure in a mathematical way, we construct the formal model 
answering the question ``Can the 
owner of~$Q$ succeed in a considered attack--defense scenario?''
The idea is to assign the Boolean 
value~$1$ to all subtrees in which the owner of~$Q$ can succeed and the value 
$0$ to the subtrees in which he cannot. Formally, we evaluate an attribute 
that we denote by~$\sat_{\owner}$, defined as follows.  First we set the 
basic assignment 
\begin{equation}\label{eq:ba_pruning}
\beta_{\sat_{\owner}}(b)=
\begin{cases}
1 & \text{if\ } b\in\mathbb{B}^{\owner}\\
0 & \text{if\ } b\in\mathbb{B}^{\negate{\owner}}.\\
\end{cases}
\end{equation}
Then, given an ADTerm $t$,  we use the following attribute 
domain\footnote{Note that the question ``Can the owner of~$Q$ succeed in the 
scenario?''
also falls into Class~$1$, as it is referring to a specific player. 
This explains why the corresponding attribute domain 
conforms to the template deduced in 
Section~\ref{sec:reduction_class_1}.
} 
to derive the values of the attribute~$\sat_{\owner}$ at all subterms of~$t$:
\begin{equation}\label{eq:AD_sat_owner}
A_{\sat_{\owner}} =
\begin{cases}
(\{0,1\}, \vee,\wedge, \wedge, \vee, \wedge,\vee) & \text{if\ } \owner=\pro\\
(\{0,1\}, \wedge, \vee,\vee, \wedge,\vee, \wedge) & \text{if\ } \owner=\opp.
\end{cases}
\end{equation}

The following theorem shows that~$\sat_{\owner}$ models the pruning procedure 
soundly and correctly. 
\begin{theorem}\label{th:equiv_pruning}
Consider a question~$Q$ of Class~$1$, its owner~$\owner$, an ADTree~$T$ and 
the corresponding ADTerm~$t$. Furthermore, let~$A_{\sat_{\owner}}$ and 
$\beta_{\sat_{\owner}}$ be defined by equations~\eqref{eq:ba_pruning}
and~\eqref{eq:AD_sat_owner}. The intuitive pruning procedure presented in 
Section~\ref{sec:pruning} removes a subtree~$T'$ of~$T$ if and only 
if the evaluation of the~$\sat_{\owner}$ attribute on the corresponding 
subterm~$t'$ of~$t$ results in the value~$0$.
\end{theorem}
\begin{proof}
We need to show that 
\begin{enumerate}
\item if a subtree is removed by pruning, the evaluation of~$\sat_{\owner}$ 
on the corresponding term results in~$0$. \label{case:removed_0}
\item if a subtree is not removed by pruning, evaluation of~$\sat_{\owner}$ 
on the corresponding term results in~$1$. \label{case:non-removed_1}
\end{enumerate}

\noindent
\ref{case:removed_0}) 
Let~$u$ be a leaf of the non-owner, from which we start the current step of 
the 
pruning procedure. We show that if a tree rooted in a node~$v$ is removed 
by pruning, then all subterms corresponding to subtrees rooted in the nodes on 
the path from~$u$ to~$v$ (including~$u$ and~$v$) evaluate to~$0$. 

We prove by contraposition. Assume that there exists a node~$w$ on the path 
between~$u$ and~$v$, such that the term corresponding to the tree rooted 
in~$w$ evaluates to~$1$. Moreover, let~$w$ be the first node with such 
property encountered when starting from~$u$. Note that~$w\not=u$, because the 
basic assignment~$\beta_{\sat_{\owner}}$ assigns the value~$0$ to every 
non-refined node of the non-owner. This means that there exists a node~$w_1$ 
which is a child of~$w$ lying on the path from~$u$ to~$v$. According to our 
assumptions, the term corresponding to the tree rooted in~$w$ evaluates to~$1$ 
while the term corresponding to the subtree rooted in~$w_1$ evaluates to~$0$. 
This implies that 
operator~$\vee$ has been used. According to the attribute domain given 
by~\eqref{eq:AD_sat_owner}, there are only three situations where the logical 
disjunction is used:
\begin{itemize}
\item either~$w$ is a properly, disjunctively refined node of the owner;
\item or~$w$ is a properly, conjunctively refined node of the non-owner;
\item or~$w$ is a refined node of the non-owner and~$w_1$ is its 
countermeasure.
\end{itemize}
It is now sufficient to notice that in all the three cases, the pruning 
procedure should have had stopped at node~$w_1$. Contradiction. 

\noindent
\ref{case:non-removed_1})
First, let us remark that the pruning procedure stops at node~$v$ if the value 
of the term corresponding to the tree rooted in the parent node of~$v$ is not 
uniquely determined by the value of the subterm corresponding to the tree 
rooted in~$v$. This is because, in all three cases where pruning stops 
at~$v$, the calculation of the~$\sat_{\owner}$ attribute for the subterm 
corresponding to the tree rooted in the parent of~$v$ uses operator~$\vee$ 
which is applied to the value~$0$ (quantifying the term corresponding to the 
tree rooted in~$v$) and another value which cannot be deduced from the
currently considered path. The tree rooted in the parent of~$v$ will either be 
removed by the pruning procedure starting from another leaf of the non-owner 
or it will not be removed after all possible steps of the pruning are 
performed.

Let~$T$ be an ADTree and~$T'$ be its subtree which is not removed by any step 
of the pruning procedure. In the remaining part of this proof we show that 
the evaluation of~$\sat_{\owner}$ on a term~$t'$ corresponding to~$T'$ results 
in value~$1$. The proof is by induction on the structure of~$T'$.

If~$T'$ is a leaf of~$T$, then it needs to represent a basic action of the 
owner. This is because all leaves of the non-owner are removed by 
pruning. According to the basic assignment~$\beta_{\sat_{\owner}}$ the 
term~$t'$ is quantified with~$1$. 

Let us now consider a tree~$T'$ which has not been removed by any step of the 
pruning procedure and which is not a leaf of~$T$. As induction hypothesis, we 
assume that the evaluation of~$\sat_{\owner}$ on all subterms corresponding to 
the subtrees of~$T'$ not removed by pruning results in~$1$. This implies that 
the evaluation of~$\sat_{\owner}$ on~$t'$ yields~$1$, because the only 
possible ways of combining the values quantifying the subterms of the 
considered term are~$\vee$ or~$\wedge$.
\end{proof}

Next, we show how to combine the evaluation of an attribute from Class~$1$ 
with pruning, in one procedure. 

\subsection{Merging Evaluation of Attributes of Class 1 With Pruning}

We have argued that, in order to evaluate an attribute~$\alpha$ of Class~$1$ 
in a correct way, we first need to prune a considered ADTree with respect to 
the owner of the corresponding question. In this section, we show how the two 
procedures of attribute evaluation and pruning can be modeled using an 
extended attribute domain. 

Consider a question $Q$ of Class~$1$, the corresponding attribute domain
$A_{\alpha}=(D, \circ, \bullet, \bullet,\circ, \bullet, \circ)$ and 
a basic assignment $\beta_{\alpha}\colon \mathbb{B}\to D$. 
For ease of presentation, in this section we assume that the owner of $Q$ is 
the proponent, 
i.e., that $\circ$ is the \emph{at least one} operator and 
$\bullet$ is the \emph{all} operator. 
In order to be able to answer~$Q$ without the necessity of first 
pruning the ADTree, we extend~$D$ with an additional Boolean dimension that 
represents which actions are relevant for our considerations. Therefore, 
instead of the value domain~$D$, we are using the Cartesian 
product~$D\times\{0,1\}$ denoted by $\widehat{D}$. Furthermore, we 
define~$\widehat{\circ}$ and~$\widehat{\bullet}$ as two internal operations on 
$\widehat{D}$, by setting
\[ 
\widehat{\circ}((d_1,s_1),\dots,(d_k,s_k))= (\circ(d_1\otimes 
s_1,\dots,d_k\otimes s_k),\!\bigvee_{i=1}^{k}s_i) \]
and
\[ 
\widehat{\bullet}((d_1,s_1),\dots,(d_k,s_k))= (\bullet(d_1\otimes
s_1,\dots,d_k\otimes s_k),\!
\bigwedge_{i=1}^{k}s_i),\]
where, for all~$d\in D$, we set~$d\otimes 1=d$ and~$d\otimes 0=e_{\circ}$,
and where $e_{\circ}$ denotes the neutral element with respect to $\circ$. 
In order to define the 
extended basic assignment~$\widehat{\beta_{\alpha}}\colon \mathbb{B}\to 
D\times\{0,1\}$, we set~$\widehat{\beta_{\alpha}}(b)=(\beta_{\alpha}(b),1)$, 
for every basic action~$b$ of the owner of~$Q$, 
and~$\widehat{\beta_{\alpha}}(b)=(\beta_{\alpha}(b),0)$, for every basic 
action~$b$ of the non-owner of~$Q$. Theorem~\ref{th:eval+pruning} shows that 
the attribute domain defined by 
\[\widehat{A_{\alpha}}=
(\widehat{D}, \widehat{\circ}, \widehat{\bullet}, 
\widehat{\bullet}, \widehat{\circ}, \widehat{\bullet}, \widehat{\circ}),
\]
constitutes a formal model allowing us to correctly 
evaluate attribute~$\alpha$, and thus answer $Q$, without requiring any prior 
pruning. 

\begin{theorem}\label{th:eval+pruning}
Let $Q$, $A_{\alpha}$, $\beta_{\alpha}$, $\widehat{A_{\alpha}}$ 
and $\widehat{\beta_{\alpha}}$
be as defined in this section. 
For every ADTerm~$t$, we have 
\begin{itemize}
\item $\widehat{\alpha}(t)=(d,0)$, for some~$d\in D$, 
if  the tree corresponding to~$t$ is 
removed by the pruning procedure related to~$Q$;
\item $\widehat{\alpha}(t)=
(\alpha(t),1)$, if  the tree corresponding to~$t$ is not 
removed by the pruning procedure related to~$Q$.
\end{itemize}
\end{theorem}
\begin{proof}
First observe that the calculation of the second component of the 
pair~$\widehat{\alpha}(t)$ corresponds to the calculation of 
attribute~$\sat_{\owner}$ formalized in Section~\ref{sec:pruning}. 
Theorem~\ref{th:equiv_pruning} ensures that the second component of 
$\widehat{\alpha}(t)$ is~$0$ if and only if the corresponding ADTerm 
is removed by the pruning procedure related to~$Q$. 
In this case, the first component of~$\widehat{\alpha}(t)$
does not have any conclusive meaning. This corresponds to the fact that 
answering the question~$Q$ for the pruned subtrees of an ADTree does not make 
any sense, since these subtrees do not contribute to the success
of the owner of~$Q$. 

Let~$t$ be a term corresponding to a subtree
which is not removed by pruning related to~$Q$. 
In order to prove that~$\widehat{\alpha}(t)=(\alpha(t),1)$, 
it is sufficient to notice that,
according to Theorem~\ref{th:equiv_pruning}, the evaluation 
of~$\sat_{\owner}$ on all subterms of~$t$ results in the value~$1$. 
This means that, when calculating~$\widehat{\alpha}(t)$
we perform operations of the form
\begin{align*}
\widehat{\diamond}((d_1,1),\dots,(d_k,1))& =(\diamond(d_1\otimes 
1,\dots,d_k\otimes 1), 1)\\
& =(\diamond(d_1,\dots,d_k),1), 
\end{align*}
where~$\diamond\in\{\circ,\bullet\}$. This obviously leads to the desired 
result. 
\end{proof}

We illustrate the use of the extended attribute domain introduced in this 
section on the following example. 
\begin{example}
As in Example~\ref{ex:cost_pruned}, we would like to answer 
the question 
``What are the minimal costs of the proponent, assuming 
that reusing tools is infeasible?'',
on the tree in the left of Figure~\ref{fig:for_pruning}.
From Example~\ref{ex:cost_pruned},  we know that the corresponding 
attribute domain is 
$A_{\cost}=(\mathbb{R}, \min, +, +,\min, +,\min)$.
We extend $A_{\cost}$ to the attribute domain
$\widehat{A_{\cost}}=
(\widehat{\mathbb{R}}, \widehat{\min}, \widehat{+}, 
\widehat{+},\widehat{\min}, \widehat{+},\widehat{\min})$,
 as defined in this section. 
Since $+\infty$ is the neutral element with respect to $\min$ on $\reals$, 
operation $\otimes$ is defined 
as $x\otimes 0=+\infty$, for every $x\in\reals$. 
We evaluate the minimal costs attribute
on the ADTerm corresponding to the non-pruned ADTree from 
Figure~\ref{fig:for_pruning},
as follows:
\begin{align*}
&\widehat{\cost}(\vee^{\pro}(\wedge^{\pro}(\text{IC},\text{UC}),\text{SS},
\counter^{\pro}(\text{OA},\text{FW})))=\\
&\widehat{\min}(\widehat{+}(\widehat{\beta_{\cost}}(\text{IC}),
(\widehat{\beta_{\cost}}(\text{UC}))),\widehat{\beta_{\cost}}(\text{SS}),
\widehat{+}(\widehat{\beta_{\cost}}(\text{OA}),\widehat{\beta_{\cost}}(\text{FW})))=\\
&\widehat{\min}(\widehat{+}((\beta_{\cost}(\text{IC}),1),
(\beta_{\cost}(\text{UC}),1)),(\beta_{\cost}(\text{SS}),1),
\widehat{+}((\beta_{\cost}(\text{OA}),1),(\beta_{\cost}(\text{FW}),0)))=\\
&\widehat{\min}(
(+(\beta_{\cost}(\text{IC}),\beta_{\cost}(\text{UC})),1),
(\beta_{\cost}(\text{SS}),1),
(+(\beta_{\cost}(\text{OA}),+\infty),0))
=\\
&\widehat{\min}(
(+(100\text{\officialeuro},200\text{\officialeuro}),1),
(400\text{\officialeuro},1),
(+\infty,0))
=\\
&\widehat{\min}(
(300\text{\officialeuro},1),
(400\text{\officialeuro},1),
(+\infty,0)
)=\\
&(\min\{
300\text{\officialeuro},
400\text{\officialeuro},
+\infty
\},1)=\\
&(
300\text{\officialeuro},1).
\end{align*}
This result shows that the scenario is satisfiable for the proponent 
and that his minimal costs are $300\text{\officialeuro}$. 
It is the same as the result obtained in 
Example~\ref{ex:cost_pruned}.
\end{example}

\subsection{From a Question to an Attribute Domain}
\label{sec:q_to_ad_class_1}

In this section we analyze how a question of Class~$1$ should
look like, in order to be able to instantiate the attribute domain template
$A=(D,\circ,\bullet, \bullet,\circ, \bullet,\circ)$ 
with specific value set and operators. 
To correctly instantiate $A$, we need a value domain $D$, two operators 
(for \emph{all} and \emph{at least one}) and we need to know 
which of those operators instantiates  $\circ$ and which $\bullet$. Thus, a well 
specified question of Class~$1$ contains exactly four parts, as illustrated  
on the following question:

\medskip
\begin{tabular}{lll}
\emph{Modality:} & & What are the minimal\\
\emph{Notion:} & & costs \\
\emph{Owner:} & & of the proponent \\
\emph{Execution:} & & assuming that all actions are executed one after 
another? 
\end{tabular}

\medskip
\noindent
Each of the four parts  has a specific purpose in determining the attribute 
domain. 

\noindent\emph{\textbf{Notion.}} The notion used by the question influences
the choice of the value domain. 
The notions in Class~$1$, identified during our study, are:
\begin{center}
\begin{minipage}[t]{3.9cm}
\begin{itemize}
\item attack potential,
\item attack time,
\item consequence,
\item costs,
\item detectability,
\item difficulty level,
\item elapsed time,
\end{itemize}
\end{minipage} 
\begin{minipage}[t]{3.9cm}
\begin{itemize}
\item impact,
\item insider required,
\item mitigation success,
\item outcome,
\item penalty,
\item profit,
\item response time,
\end{itemize}
\end{minipage}
\begin{minipage}[t]{3.9cm}
\begin{itemize}
\item resources,
\item severity,
\item skill level,
\item special equipment\newline needed,
\item special skill needed,
\item survivability.
\end{itemize}
\end{minipage}
\end{center}
\noindent From the notion we determine the value domain, e.g.,
$\nats$, $\reals$, $\getz$, etc. The choice of the value domain influences 
the basic assignments, as well as the operators determined by the 
modality and the execution style. The selected value domain 
needs to include all values that we want to use 
to quantify the owner's actions. It also must contain a neutral element 
with respect to 
$\circ$, if $\owner=\pro$, and with respect to $\bullet$, if $\owner=\opp$. 
This neutral element is assigned to all 
non-refined nodes of the non-owner, as argued in 
Section~\ref{sec:pruning}.

\noindent\emph{\textbf{Modality.}} The modality of a question clarifies 
how options are treated. Thus, it determines the characteristic of the
\emph{at 
least one} operator. Different notions are accompanied with different 
modalities. In the case of costs, interesting modalities are 
minimal, maximal and average. 

\noindent\emph{\textbf{Execution.}} The question also needs to 
specify an execution style. Its value determines the treatment when all 
actions need to be executed. Thus, it describes the characteristic of the
\emph{all} operator. Exemplary execution styles are: 
 simultaneously/sequentially (for time) or with 
 reuse/without reuse (for resources). 

\noindent\emph{\textbf{Owner.}} The owner of a question determines how 
the modality and the execution are mapped to~$\circ$ 
and~$\bullet$. In case the owner of the question is the root 
player, i.e., the proponent,~$\circ$ is instantiated with the \emph{at least one} operator 
and~$\bullet$ with the \emph{all} operator. In case the root player is not the 
owner, the instantiations are reciprocal. 

Given all four parts, we can then construct the appropriate attribute domain. 
For the notion of continuous time, also called duration, possible 
combinations of the modality, the execution style and the 
owner have been determined in Table~\ref{tab:determination_structure}. 
We instantiate the attribute domain template 
$(D,\circ,\bullet,\bullet,\circ,\bullet,\circ)$  with 
the elements of the algebraic structure $(D,\circ,\bullet)$, and use the 
value indicated in the last column of the table as the 
basic assignment for all non-refined nodes of the non-owner. The table can be 
used in the case of other notions as well, as shown in the next example. 
%
%
\begin{table*}[tbh]
\centering
\begin{tabular}{|p{0.5cm}|p{1.3cm}|p{1.4cm}|p{1.0cm}|p{1.6cm}|
p{2.8cm}|p{3.5cm}|}
\hline  & Notion & Modality & Owner & Execution 
& Structure $(D,\circ,\bullet)$ & 
Basic assignment for~$\negate\owner$ 
\\ 
\hline 1 & duration &  min & $\pro$ & sequential & $(\reals,\min,+)$ 
&   $+\infty$ 
\\ 
\hline 2 & duration &  avg & $\pro$ & sequential & $(\reals,\avg, +)$ 
&   $e_{\avg}$   \\
\hline 3 & duration &  max & $\pro$ & sequential & $(\reals,\max,+)$ 
&   $-\infty$  \\
\hline 4 & duration &  min & $\opp$ & sequential & $(\reals,+,\min)$ 
&   $0$   \\
\hline 5 & duration &  avg & $\opp$ & sequential & $(\reals,+,\avg)$ 
&   $0$  \\
\hline 6 & duration &  max & $\opp$ & sequential & $(\reals,+,\max)$ 
&   $0$  \\
\hline 7 & duration &  min & $\pro$ & parallel & $(\reals,\min,\max)$ 
&   $+\infty$  \\
\hline 8 & duration &  avg & $\pro$ & parallel & $(\reals,\avg,\max)$ 
&   $e_{\avg}$   \\
\hline 9 & duration &  max & $\pro$ & parallel & $(\reals,\max,\max)$ 
&   $-\infty$   \\
\hline 10 & duration &  min & $\opp$ & parallel & $(\reals,\max,\min)$ 
&   $-\infty$  \\
\hline 11 & duration &  avg & $\opp$ & parallel & $(\reals,\max,\avg)$ 
&   $-\infty$   \\
\hline 12 & duration &  max & $\opp$ & parallel & $(\reals,\max,\max)$ 
&   $-\infty$\\
\hline 
\end{tabular} 
\vspace*{0.5cm}
\caption{Determining instantiation of the structure in Class~$1$, where 
$e_{\avg}$ denotes the neutral element with respect to $\avg$.}
\vspace*{-0.5cm}
\label{tab:determination_structure}
\end{table*}
\begin{example}\label{ex:cost_pruned}
The question ``What are the minimal costs of the proponent, as\-suming that 
reusing tools is infeasible?'' can be answered using the attribute domain 
$A_{\cost}=(\reals, \min, +, +,\min, +,\min)$. Here the notion is 
\emph{cost}, which has the same value domain as duration, i.e., 
$\reals$. The modality is 
\emph{minimum}, the owner is \emph{the proponent} and the execution style is 
\emph{without reuse}, which corresponds to sequential. Hence, we use the 
structure   $(\reals, \min, +)$, as specified in Line~$1$ of 
Table~\ref{tab:determination_structure}. In order to answer the 
question on the tree in the left of Figure~\ref{fig:for_pruning}, we first 
prune it, as shown on the right of Figure~\ref{fig:for_pruning}. The only 
basic actions that are left are ``Internally connected'', ``User Creds'' and 
``Steal Server''. Suppose the costs are 100\officialeuro, 200\officialeuro, 
and 400\officialeuro, respectively. We use those values as basic assignment 
$\beta_{\cost}$ and apply the bottom-up computation to the 
ADTerm~$\vee^{\pro} (\wedge^{\pro}(\text{IC},\text{UC}),
\text{SS})$:
\begin{align*}
\cost(\vee^{\pro} (\wedge^{\pro}(\text{IC},\text{UC}),
\text{SS}))=& 
\vee^{\pro}_{\cost} (\wedge^{\pro}_{\cost}(\beta_{\cost}(\text{IC}),
\beta_{\cost}(\text{UC}),\beta_{\cost}(\text{SS}))=\\
&\min\{+(100\text{\officialeuro}, 
200\text{\officialeuro}), 400\text{\officialeuro}\} = 300\text{\officialeuro}.
\end{align*}
\end{example}

We would like to remark that if the structure $(D,\circ,\bullet)$
forms a semi-ring, it is not necessary to prune the ADTree to correctly answer 
a question $Q$ of Class~$1$. This is due to the fact that in a semi-ring the 
neutral element\footnote{Such an element is usually called \emph{zero} of the 
semi-ring. For instance, $+\infty$ is the zero element of the semi-ring
$(\reals,\min,+)$.} for the first operator is at the same time absorbing for 
the 
second operator. 
Such element can then be assigned to all subtrees which do 
not yield a successful scenario for the owner of $Q$, in particular to the 
uncountered basic actions of the non-owner. 



\section{Questions Where Answers for Both Players Can Be Deduced From Each 
Other}
\label{sec:winning}

%

We illustrate the construction of the attribute domain for 
Class~$2$ using the question ``What is the success probability of a scenario,
assuming that all actions are independent?''
In case of questions of Class~$2$, values assigned to a subtree 
quantify the considered property from the point of view 
of the root player of the subtree. 
This means that, if a subtree rooted in an attack node is assigned the value 
$0.2$, 
the corresponding \emph{attack is successful} with probability $0.2$.
If a subtree rooted in a defense node is assigned the value $0.2$, 
the corresponding \emph{defensive measure is successful} with probability 
$0.2$.
Thus, in Class~$2$, conjunctive and disjunctive refinements for the 
proponent and the opponent have to be treated in the same way: 
in both cases, they 
refer to the \emph{at least one} option 
(here modeled with $\circ$) 
and the \emph{all} options (modeled with $\bullet$), of the 
player whose node is currently considered.

Questions in Class~$2$ have the property that, given a value for one 
player, we can immediately deduce a corresponding value for the other player. 
For example, if the attacker succeeds with  probability $0.2$ the 
defender succeeds with probability $0.8$. This property is modeled using
a value domain with a predefined unary negation 
operation~$\overline{\phantom{\bullet}}$. 
Negation allows us to express the operators for both countermeasures
using the \emph{all} operator where the second argument is 
negated, which we represent by $\overline{\bullet}$.
Formally, $\overline{\bullet}(x,y)=x \bullet \overline{y}$.
Hence attribute domains of Class~$2$ 
follow the template
$(D,\circ,\bullet,\circ,\bullet,\overline{\bullet},\overline{\bullet})$.



Below we discuss three aspects that questions in Class~$2$ need to address.

\noindent\emph{\textbf{Notion.}}
Questions of Class~$2$ refer to notions for which the 
value domains contain a unary negation operation. This allows us to 
transform values of one player into values of the other player. 
Identified notions for Class~$2$ are:
\begin{center}
\begin{minipage}[t]{3.3cm}
\begin{itemize}
\item feasibility,
\item satisfiability,
\end{itemize}
\end{minipage}
\begin{minipage}[t]{4.5cm}
\begin{itemize}
\item probability of success, 
\item probability of occurrence,
\end{itemize}
\end{minipage}
\begin{minipage}[t]{3.9cm}
\begin{itemize}
\item needs electricity.
\end{itemize}
\end{minipage} 
\end{center}
\noindent\emph{\textbf{Modality.}} Modality specifies the operator 
for \emph{at least one} option. 
For the notions enumerated above, 
this will either be the logical OR ($\vee$) or the probabilistic 
addition of independent events $P_{\cup}(A,B) = P(A) + P (B) - P(A)P(B)$, for a 
given probability distribution $P$ and events $A$ and $B$. 

\noindent\emph{\textbf{Execution.}} 
Finally, we need to know what is the execution style, so that we can 
specify the operator 
for \emph{all} 
options. In the above notions, this will either be the logical AND 
($\wedge$) or the probabilistic 
multiplication of independent events $P_{\cap}(A,B) = P(A)P(B)$.

%
\begin{example}
We calculate
the success probability of the scenario given in 
Figure~\ref{fig:main},
assuming that all actions are independent. 
First we set the success probability of all basic actions to $\beta_{\pb} = 0.
4$
and then we use the attribute domain 
$A_{\pb}=([0,1],P_{\cup},P_{\cap},P_{\cup},P_{\cap},\overline{P_{\cap}},
\overline{P_{\cap}})$, 
where $\overline{P_{\cap}}(A,B) = P_{\cap}(A,\overline{B})$ to compute 
\begin{align*}
& P_{\cup}(P_{\cap}(\beta_{\pb}(\text{IC}),\beta_{\pb}(\text{UC})), 
\beta(\text{SS}),P_{\cap}(\beta_{\pb}(\text{OA}),1-\beta_{\pb}(\text{FW}))) =
\\ 
 &  P_{\cup}(P_{\cap}(0.4,0.4), 0.4,P_{\cap}(0.4,1-0.4))= 
 P_{\cup}(0.16, 0.4, 0.24) = 0.61696.
\end{align*}
\end{example}


\section{Questions Relating to an Outside Third Party}
\label{sec:environmental}
%

Suppose an outsider is interested in the overall maximal power consumption 
of the scenario. As in the previous section, disjunctive refinements of both 
players should be treated with one operator and conjunctive refinements of both 
players with another operator. Indeed, for a third party 
the important information is whether 
\emph{all} or \emph{at least one} option need 
to be executed and not who performs the actions. 
Also countermeasures lose their opposing aspect and their values 
are  aggregated in the same way as conjunctive refinements. Regarding 
the question, this is plausible since both the countered and the countering 
action contribute to the overall power consumption. These 
observations result in the following template for an attribute domain in 
Class~$3$: $(D,\circ,\bullet,\circ,\bullet,\bullet,\bullet)$.

We specify relevant parts of the questions in Class~$3$ on the following example.

\medskip
\begin{tabular}{lll}
\emph{Modality:} & & What is the maximal \\
\emph{Notion:} & & energy consumption \\
\emph{Execution:} & & knowing that sharing of power is impossible?
\end{tabular}
\medskip

\noindent\emph{\textbf{Notion.}} In Class~$3$, we use notions 
that express universal properties covering both players. 
Found examples are:
\begin{center}
\begin{minipage}[t]{3.2cm}
\begin{itemize}
\item social costs,
\item global costs, 
\item third party costs,
\end{itemize}
\end{minipage} 
\begin{minipage}[t]{4.2cm}
\begin{itemize}
\item environmental costs,
\item environmental damage,
\item information flow,
\end{itemize}
\end{minipage} 
\begin{minipage}[t]{4.5cm}
\begin{itemize}
\item combined execution time,
\item required network traffic,
\item energy consumption.
\end{itemize}
\end{minipage}
\end{center}
\noindent\emph{\textbf{Modality.}} 
The question should also contain enough information 
to allow us to specify how to deal with \emph{at least one} option. In general, modalities 
used in Class~$3$ are the same as those in Class~$1$, e.g., 
minimal, maximal and average.

\noindent\emph{\textbf{Execution.}} 
Finally, we need to know what is the execution 
style, so that we can  define the correct operator for \emph{all} options. 
The choices for execution style in Class~$3$ are again the same as in 
Class~$1$.

These three parts now straightforwardly define an algebraic structure 
$(D,\circ,\bullet)$ that we use to construct the attribute domain 
$(D,\circ,\bullet,\circ,\bullet,\bullet,\bullet)$.
\begin{example}
Consider the question ``What is the maximal energy consumption for the scenario 
depicted in Figure~\ref{fig:main}, knowing that sharing of power is 
impossible?''
Both, the proponent's as well as the opponent's actions may require energy. 
We assume that being ``Internally 
Connected'', performing an 
``Outsider Attack'' and running a ``Firewall'' all consume 20\si{kWh}.
Obtaining ``User Creds'' requires 1\si{kWh}, whereas ``Stealing Server'' does
not require any energy. These numbers constitute the basic assignment 
for the considered attribute. From the question we know that, 
when we have a choice, we should consider the option which consumes the most 
energy. Furthermore, since sharing of power is impossible, 
values for actions which require execution of several subactions should be 
added.
Thus, we use the attribute domain 
$A_{\energy_{\max}}=(\reals, \max,+,\max,+,+,+)$ and compute the maximal 
possible energy consumption in 
the scenario as
\begin{align*}
&\energy_{\max}((\vee^{\pro} (\wedge^{\pro}(\text{IC},\text{UC}),
\text{SS}))=\\
&\max\{+(20\si{kWh},1\si{kWh}),0\si{kWh},+(20\si{kWh},20\si{kWh})\}= 
40\si{kWh}.
\end{align*} 
\end{example}

Due to similarities for modality and execution style 
for questions of Class~$1$ and Class~$3$, 
we can make use of Table~\ref{tab:determination_structure}, to choose the 
structure $(D,\circ,\bullet)$ that determines an attribute domain 
for a question of Class~$3$. The table corresponds to the case where
the owner is the proponent.  

\section{Methodological Advancements for Attack Trees}
\label{sec:attack_trees}
ADTrees extend the well-known formalism of attack trees~\cite{Schn} by 
incorporating defensive measures to the model. Hence, 
every attack tree is in particular an ADTree. 
As visible in Example~\ref{eg:imprecise}, 
underspecified questions are not a new 
phenomenon of ADTrees, but already occur in the case of pure attack trees. 
Thus, the formalization of quantitative questions, proposed in this paper,
is not only useful in the attack--defense tree methodology but, 
more importantly, it 
helps users of the more widely spread formalism of attack trees. 

Given a well specified question on ADTrees and the corresponding 
attribute domain, we can answer the question 
on attack trees as well. 
Formally, attack trees are represented with terms involving only operators 
$\vee^{\pro}$ and $\wedge^{\pro}$.
If 
$A_{\alpha}=(D_{\alpha}, \vee^{\pro}_{\alpha}, \wedge^{\pro}_{\alpha},
\vee^{\opp}_{\alpha}, \wedge^{\opp}_{\alpha}, \counter^{\pro}_{\alpha}, 
\counter^{\opp}_{\alpha})$ is an attribute domain for ADTerms, 
the corresponding attribute domain for attack trees
is $A_{\alpha}=(D_{\alpha}, \vee^{\pro}_{\alpha}, \wedge^{\pro}_{\alpha})$, 
which corresponds to the formalization introduced in~\cite{MaOo}.
Furthermore, 
due to the fact that 
attack trees involve only one player
(the attacker), the notions of attacker, proponent, 
and question's owner coincide in this simplified model. 
This in turn implies that, 
in the case of attack trees, 
the three classes of questions considered in this paper 
form in fact one class. 

\section{Prototype Tool}
\label{sec:tool}
In order to automate the analysis of security scenarios using the 
attack--defense methodology, we have developed  a prototype 
software tool, called  \emph{ADTool}. It is written in Java and is 
compatible with multiple platforms 
(Windows, Linux, MAC OS). ADTool is publicly available~\cite{ADTool}. Its main 
functionalities include the possibility of 
creation and modification of ADTree and ADTerm models as well as attributes 
evaluation on ADTrees.

ADTool combines the features offered by graphical tree representations with 
mathematical functionalities provided by ADTerms and attributes. The user can
choose whether to work with intuitive ADTrees or with formal ADTerms. 
When one of these models is created or modified, the other one is generated 
automatically. The possibility of modular display of ADTrees makes ADTool 
suitable for dealing with large industrial case studies which may correspond 
to very complex scenarios and may require large models. 

The software supports attribute evaluation on ADTrees, as presented in this 
paper. A number of predefined attribute domains allow the user to answer 
questions of Classes~$1$, $2$ and $3$. Implemented attributes include:
costs, satisfiability, time and skill level, for various owners, 
modalities and execution styles; scenario's 
satisfiability and success probability;  
reachability of the root goal in less than $x$  
minutes, where $x$ can be customized by the user; 
and the maximal energy consumption.


\section{Conclusions}
\label{sec:conclusion}

A useful feature of the attack--defense tree methodology is that it combines 
an intuitive representation and algorithms with formal mathematical modeling. 
In practice we model attack--defense scenarios in a graphical way 
and we ask intuitive questions about aspects and properties that we are 
interested in. To formally analyze the scenarios, we employ attack--defense 
terms and attribute domains. In this paper, we have guided the user in how to 
properly formulate a quantitative question on an ADTree and how to then 
construct the corresponding attribute domain. Since attack trees are a 
subclass of attack--defense trees, our results also 
advance the practical use of quantitative analysis of attack trees.


We are currently applying the approach presented in this paper
to analyze socio-technical weaknesses of real-life scenarios, 
such as Internet web filtering, 
which involve trade offs between security and usability.
In the future, we also plan to investigate the relation between 
attribute domains of all 
three classes and the problem of equivalent representations of the same 
scenario, formalized in~\cite{KoMaRaSc2}. 

\paragraph{Acknowledgments:} 
We would like to thank Piotr Kordy for his contributions to the development of 
ADTool. This work was supported by the Fonds National de la Recherche 
Luxembourg under the grants C$08$/IS/$26$ and PHD-$09$-$167$.


\bibliographystyle{splncs03}

\end{document}